\documentclass[12pt,reqno]{amsart}
\usepackage{amssymb}
\usepackage{float}
\usepackage{graphicx}
\usepackage{textcomp}
\usepackage{amsthm, mathrsfs}
 \usepackage{xcolor}
\usepackage[colorlinks=true, linkcolor=red, citecolor=blue]{hyperref}
\DeclareSymbolFont{rmlargesymbols}{OMX}{mdbch}{m}{n}
\DeclareMathSymbol{\rmintop}{\mathop}{rmlargesymbols}{82}

\theoremstyle{plain}

\newtheorem{corollary}{Corollary}

\newtheorem{proposition}{Proposition}

\newtheorem{theorem}{Theorem}
\numberwithin{equation}{section}

\usepackage[top=1.7cm, bottom=1.7cm, left=2.5cm, right=2.5cm]{geometry}

\begin{document}
\title{$q$-deformed coherent states associated with the sequence $x_n^{q,\alpha }=(1+\alpha q^{n-1})[n]_q$}
\maketitle
\begin{center}
\author{Othmane El Moize $^{*}$, Zouha\"ir Mouayn $^{\flat}$ \ and
\ Khalid Ahbli $^{\sharp}$ }\\
\vspace*{2em}
\begin{scriptsize}
$^{*}$ 46 Lot El Youssr 2,  26100, Berrechid, Morocco\\
$^{\flat}$ Department of Mathematics, Faculty of Sciences and Technics (M’Ghila),\vspace*{0.2mm}\\ Sultan Moulay
Slimane University, P.O. Box. 523, Béni Mellal, Morocco\vspace*{-0.4mm}\\
$^{\sharp}$ Route principale Lqliaa, P.O. Box. 456, Inzegane Ait Melloul, Agadir, Morocco.
\end{scriptsize}
\end{center}
\begin{abstract}
We introduce  new  generalized $q$-deformed coherent states ($q$-CS) by replacing the $q$-factorial of $[n]_q!$ in the series expansion of the classical $q$-CS by the generalized factorial $x_n^{q,\alpha}!$ where $x_n^{q,\alpha}=(1+\alpha q^{n-1})[n]_q$. We use the shifted operators method based on the  sequence  $x_n^{q,\alpha}$ to obtain a  realization in terms of  Al-Salam-Chihara polynomials for the basis vectors of the Fock space carrying the constructed $q$-CS. These new states interpolate between the  $q$-CS of Arik-Coon type ($\alpha=0$, $0<q<1$) and a set of  coherent states of Barut–Girardello type for the Meixner-Pollaczek oscillator ($\alpha\neq 0$, $q\to 1$). We also discus their  associated Bargmann type transforms.  
\end{abstract}
\begin{scriptsize}
\\ \medskip
\textbf{AMS CLASSIFICATION}: 33C45, 44A15, 46E22, 32C81, 05A30, 3D45. 
\end{scriptsize}
\section{Introduction}
Coherent states (CS) have been introduced by Schr\"{o}dinger as states which
$  $behave in many respects like classical states \cite{schro}. The \textit{canonical} CS of the harmonic oscillator denoted $\vartheta_z$ and labeled by points $z\in 
\mathbb{C}$, may be defined in four ways: (i) as eigenstates of the
annihilation operator $a$, (ii) by applying the displacement operator $%
e^{za^{\ast }-\bar{z}a}$ on the vacuum state $|0\rangle $ such that $%
a|0\rangle =0$, where $a^{\ast }$ is the Hermitian conjugate of $a$ satisfying the relation $[a,a^*]=1$, (iii)
by finding states that minimize the Heisenberg uncertainty principle and
(iv) as a specific superposition of eigenstates $\varphi _{n}$ of the harmonic
oscillator (HO) number operator $a^{\ast }a$ as 
\begin{equation}
\vartheta_{z}:=\left( e^{z\bar{z}}\right) ^{-1/2}\sum_{n=0}^{+\infty }\frac{\bar{z%
}^{n}}{\sqrt{n!}}\varphi _{n},  \label{CS01}
\end{equation}%
where the $\varphi_n$'s span a Hilbert space $\mathscr{H}$ usually called Fock space.

The generalized CS (GCS) may be constructed by extending one of the four
aforementionned ways defining the canonical ones or by superposing CS themselves \cite{Birula}. The terminology of GCS was first appeared and studied in \cite{Birula,Stoler} in connection with states discussed in \cite{TG65}. GCS belong to a larger class of states named \textit{nonclassical }which\textit{\ }are involved in quantum optics and in other fields
ranging from solid states to cosmology. These states exhibit some purely
quantum-mechanical properties, such as squeezing and antibunching \cite{dodo02}.

GCS are usually associated with potential algebras other than the
oscillator one \cite{KS85, pere86}. An important example is provided by the $q$%
-deformed CS ($q$-CS for brevity) related to deformations of boson
operators \cite{ACO, Bi89, Mac89}. $q$-CS are usually constructed in a way that
they reduce to their standard counterparts as $q\rightarrow 1$.  Among $q$-CS, there are those associated with the relation $a_{q}a_{q}^{\ast }-qa_{q}^{\ast }a_{q}=1$ with $0<q<1$, where  $a_{q}$ are often termed maths-type $q$-bosons operators \cite%
{Solo94, JKAS94}. These $q$-CS may be defined \cite{ACO} through a $q$-analog of the number states expansion \eqref{CS01} as 
\begin{equation}  \label{CS02}
\vartheta_{z}^{q}:=\left( e_{q}(z\bar{z})\right) ^{-\frac{1}{2}%
}\sum_{n=0}^{+\infty }\frac{\bar{z}^{n}}{\sqrt{[n]_{q}!}}\varphi _{n}^{(q)},
\end{equation}%
for $z\in\mathbb{C}_q:=\{\zeta\in\mathbb{C},\:(1-q)\zeta\bar{\zeta}<1\}$ where 
\begin{equation}\label{expodef}
e_{q}(\xi):=\sum_{n=0}^{+\infty }\frac{\xi^{n}}{[n]_{q}!}=\frac{1}{(\xi(1-q);q)_{\infty}}
\end{equation}
being a $q$-exponential function  and $[n]_q!=\displaystyle\prod_{k=1}^{n}[k]_q$ where 
\begin{equation}
[k]_q=\frac{1-q^k}{1-q}\to k \quad \mathrm{as}\quad q\to 1.
\end{equation}
Details on the basic notations of $q$-calculus can be found in \cite{KS, Ism29,GR}. Here, the $\varphi _{n}^{(q)}$'s  span a Hilbert space $\mathscr{H}_{q}$ which stands for a $q$-analog of the Fock space $\mathscr{H}$.\medskip

In this paper, we construct a new class of generalized $q$-CS by replacing  the $q$-factorial $[n]_q!$ in the denominator following the summation sign in \eqref{CS02} by the generalized factorial
\begin{equation}\label{x_n^gamma!}
x_{n}^{q,\alpha }!:=(-\alpha;q)_n[n]_q!,\quad x_{0}^{q,\alpha }=1,\quad  0< q <1,\:-1<\alpha<q.
\end{equation}
We discuss the resolution of the identity operator on $\mathscr{H}_q$ satisfied by these states in connection with a moment problem associated with the sequence of numbers $x_{n}^{q,\alpha}=(1+\alpha  q^{n-1})[n]_q$, $n=0,1,2,\cdots$. We also  give a  realization  in terms of the Al-Salam-Chihara polynomials  of the basis $\varphi _{n}^{(q)}$ by using the shift operators method \cite{AI2012} based on  the sequence $x_{n}^{q,\alpha}$. The constructed $q$-CS interpolate between the $q$-CS of Arik-Coon type \cite{SOZ18}, which correspond to  $\alpha=0$ and  $0<q<1$ and a set of  coherent states of Barut-Girardello type for the Meixner-Pollaczek oscillator \cite{Bor}, which we recover at the limit $q\to 1$ with $\alpha\neq 0$. We also discuss the  associated coherent states transforms and some particular cases of parameters $q$ and $\alpha$.

The paper is organized as follows. In Section 2, we  introduce a new class of  generalized $q$-CS for which we prove the resolution of the identity  property. Section 3 is devoted to  a polynomial realization of the basis $\{\varphi_n^{(q)}\}_{n=0}^\infty$. In Section 4, we define the  Bargmann type  transform associated with our  $q$-CS.

\section{A class of generalized $q$-CS}
Let  $q,\, \alpha$ such that  $0< q <1$ and $-1<\alpha<q$ be fixed parameters. We define a set of G$q$-CS associated with the sequence $x_{n}^{q,\alpha }$ via the superposition 
\begin{equation}
\vartheta_{z}^{q,\alpha} :=\left( \mathcal{N}_{q,\alpha }(z\bar{z}%
)\right) ^{-\frac{1}{2}}\sum_{n=0}^{+\infty }\frac{\bar{z}^{n}}{\sqrt{%
x_{n}^{q,\alpha }!}} \varphi _{n}^{(q)} ,\ \ n=0,1,2,...\text{ },
\label{NLCS_x_n,alpha,q}
\end{equation}
where  $\varphi _{n}^{(q)}$ are as in \eqref{CS02}. We obtain, by direct calculations,  the normalization factor in \eqref{NLCS_x_n,alpha,q} as  
\begin{equation}\label{norfac}
\mathcal{N}_{q,\alpha}\left( r^2\right)={}_2\phi_1\left(\begin{array}{c}0,0\\
 -\alpha\end{array}\Big|q; r^2(1-q)\right),\quad r^2=z\bar{z},
\end{equation}
where 
\begin{equation}
  \setlength\arraycolsep{1pt}
{}_m \phi_s\left(\begin{matrix}a_1,...,a_m \\ b_1,...,b_s \end{matrix}\left|q;\xi\right.\right):=\displaystyle\sum_{k=0}^{\infty} \frac{(a_1,...,a_m;q)_k}{(b_1,...,b_s;q)_k}(-1)^{(1+s-m)k}q^{(1+s-m)\binom{k}{2}}\frac{\xi^k}{(q;q)_k}
  \end{equation}
  is the basic hypergeometric series (\cite{KS}, p.11). The expression \eqref{norfac} is finite for    $z\in\mathbb{C}_{q}:=\{\xi\in\mathbb{C},\xi\bar{\xi}<(1-q)^{-1}\}$. More precisely, we show the following crucial property.
\begin{proposition}
For $(q,\, \alpha)$ such that   $0< q <1$ and $-1<\alpha<q$, the states \eqref{NLCS_x_n,alpha,q} satisfy the resolution of the identity operator on $\mathcal{H}$ as 
\begin{equation}
\mathbf{1}_{\mathcal{H}}=\int_{\mathbb{C}}|\vartheta_{z}^{q,\alpha}\rangle \langle \vartheta_{z}^{q,\alpha} | \mu _{q,\alpha }(d^2z),  \label{RI_x_n^gamma}
\end{equation}
with respect to the measure
\begin{equation}\label{RI_mes}
\mu_{q,\alpha }(d^2z)=\frac{1}{2\pi}\mathcal{N}_{q,\alpha }(r^2)\rho_{q,\alpha }(dr)\lambda_{[0,2\pi)}(d\theta),\: r\geq 0,\: \theta\in[0,2\pi),
\end{equation}
where $\lambda_{[0,2\pi)}$ is the Lebesgue measure on $[0,2\pi)$ and $\rho_{q,\alpha }$ is the radial measure given by 
\begin{equation}\label{MRI}
\rho_{q,\alpha }=(-\alpha,q;q)_{\infty}\sum_{n=0}^{\infty}\frac{q^n}{(q;q)_n}h_n(-\alpha q^{-1}|q)\delta_{(1-q)^{-\frac{1}{2}}q^\frac{n}{2}},
\end{equation}
in terms of the Rogers-Szeg\"o polynomials $h_n(z|q)$.
\end{proposition}
\begin{proof}
Let us assume that the measure takes the form $\mu_{\alpha, q }(d^2z)=\frac{1}{2\pi}\mathcal{N}_{q,\alpha }(r^2)\rho_{q,\alpha }(dr)\lambda_{[0,2\pi)}(d\theta)$,  $z=r e^{i\theta},$ $r\geq 0$, $\theta\in[0,2\pi),$
where $\rho_{q,\alpha }$ is a positive radial measure to be determined. Using the expression of the states \eqref{NLCS_x_n,alpha,q}, the operator
\begin{equation}
\mathcal{O}=\int_{\mathbb{C}}|\vartheta_{z}^{q,\alpha}\rangle \langle \vartheta_{z}^{q,\alpha} | \mu_{q,\alpha}(d^2z)
\end{equation}
 can be written as
\begin{eqnarray}
\mathcal{O}&=&\sum\limits_{n,m=0}^{+\infty} |\varphi_n^{(q)}\rangle\langle \varphi_m^{(q)}|\left(  \int_0^{+\infty} \frac{r^{n+m}}{\sqrt{(-\alpha ;q)_n[n]_q!(-\alpha ;q)_m[m]_q!}}\rho_{q,\alpha }(dr)\:\frac{1}{2\pi} \int_0^{2\pi} e^{i(n-m)\theta}\lambda_{[0,2\pi)}(d\theta)\right)  \cr
&=&\sum\limits_{n=0}^{+\infty}|\varphi_n^{(q)}\rangle\langle \varphi_n^{(q)}|\frac{1}{(-\alpha ;q)_n[n]_q!}\left(  \int_0^{+\infty} r^{2n}\rho_{q,\alpha }(dr) \right) \label{O_MP}.
\end{eqnarray}
One now has to find  $\rho_{q,\alpha}$ such that
\begin{equation}
\int_0^{+\infty}r^{2n}\rho_{q,\alpha}(dr)=(-\alpha ;q)_n[n]_q!,\quad n=0,1,2,\cdots .
\end{equation}
 A solution to this moment problem can be found  in (\cite{NMT}, Theorem 3.11) as
\begin{equation}\label{rho(r)}
\rho_{q,\alpha }=(-\alpha,q;q)_{\infty}\sum_{n=0}^{\infty}\frac{q^n}{(q;q)_n}h_n(-\alpha q^{-1}|q)\delta_{(1-q)^{-\frac{1}{2}}q^\frac{n}{2}},
\end{equation}
where   the Rogers-Szeg\"o polynomials are defined by (\cite{Ism29}, p.318):
\begin{equation}
h_n(z|q)=\sum_{k=0}^{n}\frac{(q;q)_n}{(q;q)_k(q;q)_{n-k}}z ^k.
\end{equation}
By using the radial measure in  \eqref{rho(r)}, the operator $\mathcal{O}$ reduces to $\sum\limits_{n=0}^{+\infty}|\varphi_n^{(q)}\rangle\langle \varphi_n^{(q)}|$ which gives the identity operator on $\mathcal{H}$ since $\{\varphi_n^{(q)}\}$ is an orthonormal basis of $\mathcal{H}$. This completes the proof.
\end{proof}

\hspace*{-1em}\textbf{Remark 1.} When $\alpha =0$,  the measure given by \eqref{RI_mes} reduces to 
$$\mu_{q,0}(d^2z)=\frac{1}{2\pi}\lambda_{[0,2\pi)}(d\theta)\frac{(q;q)_{\infty}}{((1-q)r^2;q)_{\infty}}\sum_{n=0}^{\infty}\frac{q^n}{(q;q)_n}\delta\left(r-\frac{q^{n/2}}{\sqrt{1-q}}\right), \: z=re^{i\theta},$$
and the corresponding   G$q$-CS turn out to be the  $q$-deformed CS introduced by Arik and Coon \cite{ACO}, \cite{VM95}.\medskip 

Let $L^{2}(\mathbb{C}_q,\nu_{q,\alpha })$ denotes the space of complex-valued functions on  $\mathbb{C}_q$ which are square integrable with respect to the measure 
\begin{equation}
\nu_{q,\alpha}(d^2z)=\frac{1}{2\pi}\lambda_{[0,2\pi)}(d\theta)\rho_{q,\alpha}(dr),z=re^{i\theta},\: r\geq 0,\: \theta\in[0,2\pi),
\end{equation}
which  for $\alpha=0$, takes the form 
\begin{equation}
\nu_{q,0}(d^2z)=\frac{1}{2\pi}\lambda_{[0,2\pi)}(d\theta) (q;q)_{\infty}\sum_{n=0}^{\infty}\frac{q^n}{(q;q)_n}\delta\left(r-\frac{q^{n/2}}{\sqrt{1-q}}\right).
\end{equation}

Note that, the closure in $L^{2}(\mathbb{C}_q,\nu_{q,\alpha })$ of the  linear span of the monomials $\left\lbrace\dfrac{z^{n}}{\sqrt{%
x_{n}^{q,\alpha }!}}\right\rbrace_{n \geq 0}$  is a Hilbert
space whose reproducing kernel is given   by
\begin{eqnarray}\label{RKdem1}
K_{q,\alpha}(z,w)&=&\sum_{j=0}^{\infty}\frac{(z\bar{w})^j}{(-\alpha;q)_j[j]_q!}=\sum_{j= 0}^{\infty} \frac{(z\bar{w})^j}{(-\alpha;q)_j}\;\frac{(1-q)j}{(q;q)_j}\cr
&=&{}_2\phi_1\left(\begin{array}{c}0,0\\
 -\alpha\end{array}\Big|q; z\bar{w}(1-q)\right)
\end{eqnarray}
and  will be denoted by $\mathcal{A}_{q,\alpha }^{2}(\mathbb{C})$. This space may be viewed as a generalization (with respect  to the parameter  $\alpha$) of the classical Ar\"{\i}k-Coon space $\mathcal{A}_{q,0}^{2}(\mathbb{C})$ which stands for  the completed space of entire functions in  $L^{2}(\mathbb{C}_q,\nu_{q,0})$ and  for which the monomials  $([n]_q!)^{-1/2}z^n$ form an orthonormal basis \cite{ACO}. 
\section{A polynomials realization of the basis $\{\varphi_n^{(q)}\}_{n=0}^\infty$ }
In this section we  use  the \textit{shift operator method} as discussed by Ali and Ismail \cite{AI2012} to obtain a polynomial realization of the basis $\{\varphi_n^{(q)}\}$. For this we may define the operators $a $ and $a^*$ on the space $\mathcal{H}$ by
\begin{equation}
  a\varphi_n ^{(q)}= \sqrt{2x_{n}^{q,\alpha }}\varphi_{n-1}^{(q)},\;\; a\varphi_0^{(q)} =0, \qquad
  a^{*}\varphi_n^{(q)} = \sqrt{2x_{n+1}^{q,\alpha }}\varphi_{n+1}^{(q)}.
\label{abst-ops}
\end{equation}
Using $a $ and $a^*$, it is  possible to identify the basis vectors $\{\varphi_n^{(q)}\}_{n=0}^\infty$  with another family of real orthogonal polynomials. We now define the operators,
\begin{equation}
  Q= \frac 1{\sqrt{2}}\; [a + a^{*} ]\; , \qquad
  P = \frac 1{i\sqrt{2}}\; [a - a^{*} ]\; ,
\label{eq:pos-mom-op}
\end{equation}
to play the roles of the  standard position and momentum operators. The operator $Q$ acts on the basis vectors $\varphi_n$ as
\begin{equation}
  Q \varphi_n^{(q)}= \sqrt{x_n^{q,\alpha }}\; \varphi_{n-1}^{(q)} + \sqrt{x_{n+1}^{q,\alpha }}\; \varphi_{n+1}^{(q)}\; .
\label{eq:pos-op-act}
\end{equation}
Here, the sum $\sum_{n=0}^\infty \dfrac 1{\sqrt{x_n^{q,\alpha } }}$ diverges, therefore, the operator $Q$ is essentially self-adjoint and hence has a unique self-adjoint extension, which we again denote by $Q$. Let $E_x , \; x\in \mathbb R$, be the spectral family of $Q$, so that,
$$ Q = \int_{-\infty}^\infty x \; dE_x \; .$$
Thus there exists a measure $\omega(dx)$ on $\mathbb R$ such that on the Hilbert space $L^2 (\mathbb R , \omega(dx))$,
$Q$ is just the operator of multiplication by $x$. Consequently, on this space, the relation (\ref{eq:pos-op-act})
takes the form
\begin{equation}
  x\varphi_n^{(q)} = \sqrt{x_n^{q,\alpha }} \varphi_{n-1}^{(q)} + \sqrt{x_{n+1}^{q,\alpha }} \varphi_{n+1}^{(q)}\; 
\label{eq:pos-op-act2}
\end{equation}
which is a two-term recursion relation, familiar from the theory of orthogonal polynomials. It follows that
$\omega(dx) = d\langle \varphi_0^{(q)}\vert E_x\vert \varphi_0^{(q)}\rangle $, and the $\varphi_n^{(q)}$ may be realized as the polynomials obtained
by orthonormalizing the sequence of monomials $1, x, x^2 , x^3 , \ldots\; , $ with respect to this measure
(using a Gram-Schmidt procedure). Let us use the notation $p_n (x)$ to write the vectors $\varphi_n^{(q)}$, when they are
so realized, as orthogonal polynomials in $L^2 (\mathbb R , \omega(dx))$. Then, 
\begin{equation}
    \langle \varphi_k^{(q)} , \varphi_n^{(q)}\rangle_{\mathcal{H}} = \int_{\mathbb R}  p_k (x)p_n (x)\; \omega(dx)= \delta_{k,n}\; .
\label{eq:poly-orthog}
\end{equation}
\begin{proposition}
A polynomial realization of the basis $\{\varphi_n^{(q)}\}_{n=0}^\infty$ is given by 
\begin{equation}
\varphi_n^{(q)}(x)=((-\alpha,q;q)_n)^{-1/2}Q_{n}\left(\frac{x\sqrt{1-q}}{2},\sqrt{\alpha},-\sqrt{\alpha}|q\right),\qquad  \label{phi_n^alpha,q}
\end{equation}
in terms of Al-Salam-Chihara polynomials
\begin{eqnarray}\label{qmeix}
Q_n(\xi;a,b|q):=\frac{(ab;q)_n}{a^n}\setlength\arraycolsep{1pt}
{}_3 \phi_2\left(\begin{matrix}q^{-n},ae^{i\theta},ae^{-i\theta} \\ab,0  \end{matrix}\left|q;q\right.\right);\quad \xi=\cos \theta.
\end{eqnarray}
\end{proposition}
\begin{proof}
By setting $P_n(x)=\sqrt{x_n^{q,\alpha }!}\;\varphi_n^{(q)}(x)$ in Eq. \eqref{eq:pos-op-act2}, we obtain the relation
 \begin{eqnarray}
xP_{n}(x)=P_{n+1}(x)+(1+\alpha q^{n-1})[n]_q P_{n-1}(x). \label{Recu_H_m_b}
\end{eqnarray}
The latter one may be  compared with  the three terms recursive relation   involving the Al-Salam-Chihara polynomials (\cite{KS}, p.80, with  $a=-b,\: \alpha=a^2$) :
\begin{equation}
x 2^{-n}Q_n(x;a,-a|q)=2^{-(n+1)}Q_{n+1}(x;a,-a|q)+2^{-(n+1)}(1-q^n)(1+a^2q^{n-1})Q_{n-1}(x;a,-a|q)
\end{equation}
which enables us  to identify  $P_{n}(x)$ as 
\begin{eqnarray}
P_n(x)&=& (1-q)^{-n/2}Q_{n}\left(\frac{x\sqrt{1-q}}{2};a,-a|q\right). 
\end{eqnarray}
This ends the proof.
\end{proof} 
We should note that the orthogonality measure $\omega_{q,\alpha }(x)dx$  for the polynomials  $\varphi_n^{(q)}(x)$, where  
\begin{equation}
\omega_{q,\alpha }(x):=\frac{(q,-\alpha;q)_{\infty}}{2\pi}\sqrt{\frac{1-q}{4-(1-q)x^2}}\left(\frac{g(x,1;q)g(x,-1;q)g(x,\sqrt{q};q)g(x,-\sqrt{q};q)}{g(x,\sqrt{\alpha};q)g(x,-\sqrt{\alpha};q)}\right),
\end{equation} 
is supported on the interval $\mathcal{I}_q=]\frac{-2}{\sqrt{1-q}},\frac{2}{\sqrt{1-q}}[$, where
\begin{equation}
g(x,\gamma;q)=\displaystyle\prod_{k=0}^{\infty}(1-\gamma x(1-q)^{1/2}q^k+\gamma^2q^{2k}).
\end{equation}
When $\alpha=0$, then Al-Salam-Chihara polynomials reduce to the continuous $q$-Hermite
polynomials $H_n(x|q)$ (\cite{KS}, p.115) whose orthogonality measure is the standard $q$-Gaussian measure on $\mathcal{I}_q$ given by
\begin{equation}
\omega_{q,0}(x)=\frac{\sqrt{1-q}}{\pi}\sin \theta \prod_{n=0}^{\infty} (1-q^n)|1-q^ne^{2i\theta}|^2 dx.
\end{equation} 
where $x\sqrt{1-q}=2\cos \theta$, $\theta\in[0,\pi]$.
\section{Bargmann-type integral transforms }

The wave functions of our G$q$-CS belongs to the Hilbert space $\mathcal{H}_{q,\alpha}=L^2(\mathcal{I}_q,\omega_{q,\alpha }(x)dx)$ and are of the form
\begin{equation}
\vartheta_{z}^{q,\alpha}(x) :=\left( \mathcal{N}_{q,\alpha }(z\bar{z})\right) ^{-1/2}\sum\limits_{n= 0}^{\infty}\frac{\overline{z}^{n}}{\sqrt{x_{n}^{q,\alpha }!}}\varphi_n^{(q)}(x),\; x\in\mathbb{R}.   \label{5.1}
\end{equation}
Recalling that $\varphi_n^{(q)}(x)$ is given by its expression in  \eqref{phi_n^alpha,q}, Eq. \eqref{5.1} takes the form 
\begin{eqnarray}\label{NCSdem1}
\vartheta_{z}^{q,\alpha}(x)&=& \left(\mathcal{N}_{q,\alpha }(z\bar{z})\right) ^{-1/2}\sum\limits_{n= 0}^{\infty}\frac{\overline{z}^{n}}{\sqrt{(-\alpha;q)_n[n]_q!}}((-\alpha,q;q)_n)^{-1/2}Q_{n}\left(\frac{x\sqrt{1-q}}{2},\sqrt{\alpha},-\sqrt{\alpha}|q\right)\cr
&=&\left(\mathcal{N}_{q,\alpha }(z\bar{z})\right) ^{-1/2}\sum\limits_{n= 0}^{\infty}\frac{(\bar{z}\sqrt{1-q})^n}{(-\alpha,q;q)_n}Q_{n}\left(\frac{x\sqrt{1-q}}{2},\sqrt{\alpha},-\sqrt{\alpha}|q\right).
\end{eqnarray} 
Next, by using the generating function (\cite{KS}, p.81):
\begin{equation}
\sum\limits_{n= 0}^{\infty}\frac{t^n}{(ab,q;q)_n}Q_n(\xi;a,b|q)=\frac{1}{(te^{i\theta};q)_{\infty}}{}_2\phi_1\left(\begin{array}{c}ae^{i\theta},be^{i\theta}\\
 ab\end{array}\Big|q;te^{-i\theta}\right),\, \xi=\cos \theta
\end{equation}
for  parameters $\xi=\frac{x\sqrt{1-q}}{2},\: a=-b=\sqrt{\alpha}$  and $t=\bar{z}\sqrt{1-q}$, we arrive at the following result.
\begin{proposition} 
For $(q,\, \alpha)$ such that   $0< q <1$ and $-1<\alpha<q$, the wave function of the
$q$-CS \eqref{5.1} can be expressed as
\begin{eqnarray}\label{CScom}
\vartheta_{z}^{q,\alpha}(x)&=& \frac{\left(\mathcal{N}_{q,\alpha }(z\bar{z})\right) ^{-1/2}}{(\bar{z}\sqrt{1-q}e^{i\theta};q)_\infty}{}_2\phi_1\left(\begin{array}{c}\sqrt{\alpha}e^{i\theta},-\sqrt{\alpha}e^{i\theta}\\
 -\alpha\end{array}\Big|q;\bar{z}\sqrt{1-q}e^{-i\theta}\right)
\end{eqnarray}
where $x=\frac{2 }{\sqrt{1-q}}\cos\,\theta$ and $\mathcal{N}_{q,\alpha } (z\bar{z})$ is given by \eqref{norfac}.
\end{proposition}
\textbf{Remark 4.1.} Note that by replacing in \eqref{CScom} $z\rightarrow z\sqrt{1-q}, \ \alpha \rightarrow -q^{2\nu}$ and $e^{i\theta}\rightarrow iq^{ix}$ with $\nu>0$ and letting  $q \to 1 $,  the G$q$-CS  $\vartheta_{z}^{q,\alpha}$   reduce to the coherent states (\cite{Bor}, p.14) :
\begin{equation}
\vartheta_{z}^{1,-q^{2\nu}}(x)=\frac{|z|^{\nu-\frac{1}{2}}}{\sqrt{\Gamma(2\nu)I_{2\nu-1}(2|z|)}} e^{i\bar{z}} {}_1F_1\left(\begin{array}{c}\nu+ix\\
 2\nu\end{array}\Big| -2i\bar{z}\right).
\end{equation}
where $I_\sigma$ denotes  the modified Bessel function (\cite{Watson}, p.77). These coherent states are of Barut–Girardello type  since they may be expanded as 
\begin{equation}
\vartheta_{z}^{1,-q^{2\nu}}(x) =\left(\Gamma(2\nu) |z|^{1-2\nu} I_{2\nu-1}(2|z|)\right)^{-1/2}\sum_{n=0}^{\infty} \dfrac{\bar{z}^n}{\sqrt{n!(2\nu)_n}}\varphi_n^{(\nu)} (x)
\end{equation}
where $\varphi_n^{(\nu)} (x):=P_n^{\nu}(x;\frac{\pi}{2})\sqrt{\frac{n!}{(2\nu)_n}}$ are  normalized Meixner-Pollaczek polynomials (\cite{KS}, p.37) and  $\left\lbrace\dfrac{z^{n}}{\sqrt{n!(2\nu)_n}}\right\rbrace_{n \geq 0}$ is an orthonormal basis of the space $\mathfrak{F}_{\nu}(\mathbb{C})$ of entire $d\mu_{\nu}-$square integrable functions, where
\begin{equation}
d\mu_{\nu}(z)=\frac{2}{\pi\Gamma(2\nu)}r^{2\nu-1}K_{1/2-\nu}(2r)rd\theta dr, \quad z=re^{i\theta}
\end{equation}
and $K_\sigma$ denotes the MacDonald function (\cite{Watson}, p.183). The space $\mathfrak{F}_{\nu}(\mathbb{C})$ have been used by Barut and Girardello \cite{Barut} to introduce  coherent states associated with non-compact groups.\medskip

The constructed G$q$-CS allows us to introduce a Bargmann-type transform as the associated coherent states transform    $\mathcal{B}_{\alpha}^{(q)}:\mathcal{H}_{q,\alpha}\rightarrow \mathcal{A}_{q,\alpha }^{2}(\mathbb{C})$ defined, as usual (see \cite{AGA}, p.27 for the general theory), by 
\begin{equation}\label{CSTdef}
\mathcal{B}_{\alpha}^{(q)}[\varphi](z)=(\mathcal{N}_{q,\alpha }(z\bar{z}))^{\tfrac{1}{2}}\langle \varphi,\,\vartheta_{z}^{q,\alpha}\rangle_{\mathcal{H}_{q,\alpha}},\;z\in \mathbb{C}_q.
\end{equation}

\begin{theorem}
For $(q,\, \alpha)$ such that   $0< q <1$ and $-1<\alpha<q$, $\mathcal{B}_{\alpha}^{(q)}:\mathcal{H}_{q,\alpha}\rightarrow \mathcal{A}_{q,\alpha }^{2}(\mathbb{C})$ is an isometric isomorphism  given by
\begin{equation}
 \mathcal{B}_{\alpha}^{(q)}[\varphi ](z)=\int_{\mathcal{I}_q}\mathcal{T}(z,\xi)\varphi(\xi)\omega_{q,\alpha }(\xi)d\xi,\quad z\in\mathbb{C}_q,\label{BAR}
\end{equation}%
with the integral kernel 
\begin{equation}
\mathcal{T}(z,\xi)=\frac{1}{(z\sqrt{1-q}e^{-i\theta};q)_\infty}{}_2\phi_1\left(\begin{array}{c}\sqrt{\alpha}e^{-i\theta},-\sqrt{\alpha}e^{-i\theta}\\
 -\alpha\end{array}\Big|q;z\sqrt{1-q}e^{i\theta}\right)
\end{equation}
and $\xi=\frac{2 }{\sqrt{1-q}}\cos\,\theta$. 
\end{theorem}
\begin{corollary}
For $(q,\, \alpha)$ such that   $0< q <1$ and $-1<\alpha<q$, the Al-Salam-Chihara polynomials have the following integral representation 
\begin{equation}
Q_n\left(\frac{\xi\sqrt{1-q}}{2};\sqrt{\alpha},-\sqrt{\alpha}|q\right)=\int_{\mathbb{C}_q} (z\sqrt{1-q})^n \mathcal{T}(z,\xi) \nu_{q,\alpha}(d^2z),\quad \xi\in \mathcal{I}_q.
\end{equation}
\end{corollary}
\begin{corollary}
For $\alpha=0$, $\mathcal{B}_{0}^{(q)}$ maps  $\mathcal{H}_{q,0}$ onto the Arik-Coon space $\mathcal{A}_{q,0}^{2}(\mathbb{C})$, as
\begin{equation}
 \mathcal{B}_{0}^{(q)}[\varphi ](z)=\int_{\mathcal{I}_q}\frac{1}{(z\sqrt{1-q}e^{-i\theta},z\sqrt{1-q}e^{i\theta};q)_\infty}\varphi(\xi)\omega_{q,0}(\xi)d\xi,\quad \xi=\frac{2 }{\sqrt{1-q}}\cos\,\theta\label{BAR0}
\end{equation}%
for every $z\in\mathbb{C}_q$. 
\end{corollary}
\textbf{Remark 4.2.} By replacing  $\xi$ by $\xi/\sqrt{2}$ in \eqref{BAR0}, $\mathcal{B}_{0}^{(q)}$ coincides with the $q$-deformed Bargmann-type transform constructed in \cite{SOZ18}.\bigskip\\
\textbf{Acknowledgements}

The authors would like to thank the Moroccan Association of Harmonic Analysis and Spectral Geometry.



\begin{thebibliography}{99}
\bibitem{schro} E.~Schr\"{o}dinger, Die Naturwissenschaften, \textbf{14} 
(1926), 664.
	
\bibitem{Birula} Z. Bialynicki-Birula, Properties of the generalized coherent state, \textit{Phys. Rev.} \textbf{173} (1968), 1207.
\bibitem{Stoler} D. Stoler, Generalized Coherent States, \textit{Phys. Rev. D.} \textbf{4} (1971), 2309.

\bibitem{TG65} U. M. Titulaer and R. J. Glauber, Correlation functions for coherent fields, \textit{Phys. Rev.} \textbf{140} (1965), 676.
\bibitem{dodo02} V. V. Dodonov, Purity-and entropy-bounded uncertainty relations for mixed quantum states, \textit{J. Opt. B: Quantum Semiclass. Opt.%
} \textbf{4} (2002), 98.
\bibitem{KS85} J. R. Klauder and B. S. Skagerstam, Coherent States Applications
in Physics and Mathematics, World Scientific, Singapore, 1985.
\bibitem{pere86} A. Perelomov, Generalized coherent states and their applications, Springer-Verlag, Berlin, 1986.
\bibitem{ACO} 
    M. Arik and D. D. Coon: Hilbert space of analytic function and generalized coherent states, \textit{J. Math. Phys}. \textbf{17} (1976), 524.
\bibitem{Bi89} L. C. Biedenharn, The quantum group {$\mathrm{SU}_q(2)$} and a {$q$}-analogue of the boson operators, \textit{J. Phys. A.} \textbf{22} (1989), 873.

\bibitem{Mac89} A. J. MacFarlane, On $q$-analogues of the quantum harmonic oscillator and the quantum group $\mathrm{SU}_q(2)$, \textit{J. Phys. A.} \textbf{22} (1989),  4581.
 

\bibitem{Solo94} A. I. Solomon, Optimal signal-to-quantum noise ratio for deformed photons, \textit{Phys. Lett. A.} \textbf{188} (1994), 215.
\bibitem{JKAS94} J. Katriel and A. I. Solomon, Nonideal lasers, nonclassical light, and deformed photon states, \textit{Phys. Rev. A} \textbf{49} (1994), 5149.
\bibitem{KS} R. Koekoek and R. Swarttouw, The Askey-scheme of hypergeometric orthogonal polynomials and its q-analogues, \textit{Reports of the Faculty of Technical Mathematics and Informatics no. 98-17, Delft University of Technology,} Delft, 1998.
\bibitem{Ism29} M. E. H. Ismail, Classical and quantum orthogonal polynomials in one variable, Encyclopedia of Mathematics and its Applications, vol. 98, Cambridge University Press, Cambridge, 2009.

\bibitem{GR} G. Gasper and M. Rahman, Basic hypergeometric series. \textit{Cambridge University Press}, Cambridge, 2004.

\bibitem{AI2012} S. T. Ali and M. E. H. Ismail,  Some orthogonal polynomials arising from
coherent states. \textit{J Phys A: Math Theor}. \textbf{45} (2012), no. 12, 125203, 16 pp.
\bibitem{SOZ18} S. Arjika, O. El Moize and Z. Mouayn, Une $q$-d\'eformation de la transformation de Bargmann vraie-polyanalytique, \textit{C.
R. Acad. Sci. Paris.} \textbf{356} (2018), 903.
\bibitem{Bor}V. V. Borzov and E. V. Damaskinsky,  Generalized coherent states for oscillators connected with Meixner and Meixner-Pollaczek polynomials. \textit{J. Math. Sci}. \textbf{136} (2006), 3564. 
\bibitem{NMT} A. Nobuhiro, B. Marek and  H. Takahiro, Radial Bargmann representation for the Fock space of type B.\textit{ J. Math. Phys}. \textbf{57 }(2016), 021702, 13 pp.
\bibitem{VM95} H. Van Leeuwen  and H. Maassen, A $q$-deformation of the Gauss distribution, \textit{J. Math. Phys.} \textbf{36} (1995), 4743.
\bibitem{Watson} G. N. Watson, Treatise on the Theory of Bessel Functions, Cambridge University Press, Cambridge, UK,
1958.
\bibitem{Barut} A. O. Barut and L. Girardello, New “coherent states” associated with non-compact groups,
\textit{Communications in Mathematical Physics}. \textbf{21} (1971),  41.
\bibitem{AGA} S. T. Ali, J. P. Antoine and J. P. Gazeau, Coherent
states, Wavelets and their Generalizations, second edition, Springer
Science+Business Media New York, 2014.






\end{thebibliography}
\end{document}